\documentclass[1p,final]{elsarticle}
\usepackage{amsfonts,color,morefloats,pslatex}
\usepackage{amssymb,amsthm, amsmath,latexsym}

\newtheorem{theorem}{Theorem}
\newtheorem{lemma}[theorem]{Lemma}

\newcommand{\ord}{{\mathrm{ord}}}

\newcommand{\gf}{{\mathrm{GF}}}

\newcommand{\PSL}{{\mathrm{PSL}}}

\newcommand{\PAut}{{\mathrm{PAut}}} 
\newcommand{\MAut}{{\mathrm{MAut}}} 
\newcommand{\GAut}{{\mathrm{Aut}}}

\newcommand{\Sym}{{\mathrm{Sym}}}

\newcommand{\Z}{\mathbb{{Z}}}

\newcommand{\bC}{{\mathbb{C}}}

\newcommand{\m}{\mathbb{M}}

\newcommand{\cP}{{\mathcal{P}}} 
\newcommand{\cB}{{\mathcal{B}}}

\newcommand{\C}{{\mathcal{C}}}
\newcommand{\calC}{{\mathcal{C}}}

\newcommand{\cQ}{{\mathcal{Q}}}
\newcommand{\cN}{{\mathcal{N}}}
\newcommand{\bN}{{\mathcal{N}}}
\newcommand{\bQ}{{\mathcal{Q}}}

\newcommand{\ba}{{\mathbf{a}}}

\newcommand{\bc}{{\mathbf{c}}}

\newcommand{\bD}{{\mathbb{D}}}

\begin{document}

\begin{frontmatter}



\title{All binary linear codes that are invariant under $\PSL_2(n)$
}

\author[cseust]{Cunsheng~Ding}
\ead{cding@ust.hk}

\author[cseust]{Hao~Liu}
\ead{hliuar@ust.hk}

\author[mtu]{Vladimir D. Tonchev}
\ead{tonchev@mtu.edu} 

\address[cseust]{Department of Computer Science and Engineering, The Hong Kong University of Science and Technology, Clear Water Bay, Kowloon,
Hong Kong} 

\address[mtu]{Department of Mathematical Sciences, Michigan Technological University, 
Houghton, Michigan 49931, USA}


\begin{abstract}
The projective special linear group $\PSL_2(n)$ is $2$-transitive for all primes $n$ and 
$3$-homogeneous for $n \equiv 3 \pmod{4}$ on the set $\{0,1, \cdots, n-1, \infty\}$. 
It is known that the extended odd-like quadratic residue codes are invariant under $\PSL_2(n)$. 
Hence, the extended quadratic residue codes hold an infinite family of $2$-designs for primes 
$n \equiv 1 \pmod{4}$, an infinite family of $3$-designs for primes $n \equiv 3 \pmod{4}$. 
To construct more $t$-designs with $t \in \{2, 3\}$, one would search for other extended 
cyclic codes over finite fields that are invariant under the action of $\PSL_2(n)$. 
The objective of this paper is to prove that the extended quadratic residue binary codes 
are the only nontrivial extended binary cyclic codes that are invariant under 
$\PSL_2(n)$.       
\end{abstract}


\end{frontmatter}

\section{Introduction}  

An $[n, \kappa, d]$ code $\C$ over $\gf(q)$ is a $\kappa$-dimensional linear subspace of 
$\gf(q)^n$ with minimum Hamming distance $d$. Trivial linear codes of length $n$ over $\gf(q)$ 
are the linear subspace consisting only of the zero vector of $\gf(q)^n$ with dimension $0$, 
the whole space $\gf(q)^n$ with dimension $n$, the subspace $\{a(1,1, \cdots, 1): a \in \gf(q)\}$ 
with dimension $1$, and the subspace 
$$
\left\{(c_0, c_1, \cdots, c_{n-1}) \in \gf(q)^n: \sum_{i=0}^{n-1} c_i=0 \right\}  
$$    
with dimension $n-1$. 

A linear code $\C$ over $\gf(q)$ is {\em cyclic} if
$(c_0,c_1, \cdots, c_{n-1}) \in \C$ implies $(c_{n-1}, c_0, c_1, \cdots, c_{n-2})$
$\in \C$.
We may identify a vector $(c_0,c_1, \cdots, c_{n-1}) \in \gf(q)^n$
with the polynomial
$$
c_0+c_1x+c_2x^2+ \cdots + c_{n-1}x^{n-1} \in \gf(q)[x]/(x^n-1).
$$
In this way, a code $\C$ of length $n$ over $\gf(q)$ always corresponds to a subset of the quotient ring
$\gf(q)[x]/(x^n-1)$.
A linear code $\C$ is cyclic if and only if the corresponding subset in $\gf(q)[x]/(x^n-1)$
is an ideal of the ring $\gf(q)[x]/(x^n-1)$.

It is well-known that every ideal of $\gf(q)[x]/(x^n-1)$ is principal. Let $\C=\langle g(x) \rangle$ be a
cyclic code, where $g(x)$ is monic and has the smallest degree among all the
generators of $\C$. Then $g(x)$ is unique and called the {\em generator polynomial,}
and $h(x)=(x^n-1)/g(x)$ is referred to as the {\em check polynomial} of $\C$. 

Given a linear code $\C$ of length $n$ over $\gf(q)$, we can extend $\C$ into another code $\overline{\C}$ of length $n+1$ over $\gf(q)$ by adding an extended coordinate, denoted by 
$\infty$, as follows: 
\begin{eqnarray}
\overline{\C}=\{(c_0,c_1, \cdots, c_{n-1}, c_{\infty}): (c_0, c_1, \cdots, c_{n-1}) \in \C \}, 
\end{eqnarray} 
where 
$$ 
c_{\infty}=-\sum_{i=0}^{n-1} c_i. 
$$
By definition, $\C$ and $\overline{\C}$ have the same dimension, but their minimum distances may 
be the same or may differ by one. 

Let $n$ be a prime and let $q$ be a prime power such that $\gcd(q, n)=1$ and $q$ is a quadratic 
residue modulo $n$. Let $m=\ord_n(q)$, and let $\alpha$ be a generator of $\gf(q^m)^*$, which 
is the multiplicative group of $\gf(q^m)$. Set $\beta=\alpha^{(q^m-1)/n}$. Then $\beta$ is a 
$n$-th primitive root of unity in $\gf(q^m)$. Define 
$$ 
g_0(x)=\sum_{i \in \bQ} (x-\beta^i) \mbox{ and } g_1(x)=\sum_{i \in \bN} (x-\beta^i), 
$$ 
where $\bQ$ and $\bN$ are the set of quadratic residues and nonresidues modulo $n$, respectively. 
It is easily seen that $g_0(x)$ and $g_1(x)$ are polynomials over $\gf(q)$ and are divisors of $x^n-1$.  
The cyclic codes of length $n$ over $\gf(q)$ with generator polynomials $g_0(x)$ and $g_1(x)$ are called \emph{odd-like 
quadratic residue codes}.  

The set of coordinate permutations that map a code $\C$ to itself forms a group, which is referred to as 
the \emph{permutation automorphism group\index{permutation automorphism group of codes}} of $\C$
and denoted by $\PAut(\C)$. If $\C$ is a code of length $n$, then $\PAut(\C)$ is a subgroup of the 
\emph{symmetric group\index{symmetric group}} $\Sym_n$. 

A \emph{monomial matrix\index{monomial matrix}} over $\gf(q)$ is a square matrix having exactly one 
nonzero element of $\gf(q)$  in each row and column. A monomial matrix $M$ can be written either in 
the form $DP$ or the form $PD_1$, where $D$ and $D_1$ are diagonal matrices and $P$ is a permutation 
matrix. 

The set of monomial matrices that map $\C$ to itself forms the group $\MAut(\C)$,  which is called the 
\emph{monomial automorphism group\index{monomial automorphism group}} of $\C$. Clearly, we have 
$$
\PAut(\C) \subseteq \MAut(\C).
$$

The \textit{automorphism group}\index{automorphism group} of $\C$, denoted by $\GAut(\C)$, is the set 
of maps of the form $M\gamma$, 
where $M$ is a monomial matrix and $\gamma$ is a field automorphism, that map $\C$ to itself. In the binary 
case, $\PAut(\C)$,  $\MAut(\C)$ and $\GAut(\C)$ are the same. If $q$ is a prime, $\MAut(\C)$ and 
$\GAut(\C)$ are identical. In general, we have 
$$ 
\PAut(\C) \subseteq \MAut(\C) \subseteq \GAut(\C). 
$$

By definition, every element in $\GAut(\C)$ is of the form $DP\gamma$, where $D$ is a diagonal matrix, 
$P$ is a permutation matrix, and $\gamma$ is an automorphism of $\gf(q)$.   
The automorphism group $\GAut(\C)$ is said to be $t$-transitive if for every pair of $t$-element ordered 
sets of coordinates, there is an element $DP\gamma$ of the automorphism group $\GAut(\C)$ such that its 
permutation part $P$ sends the first set to the second set. 
The automorphism group $\GAut(\C)$ is said to be $t$-homogeneous if for every pair of $t$-subsets of coordinates, there is an element $DP\gamma$ of the automorphism group $\GAut(\C)$ such that its 
permutation part $P$ sends the first set to the second set. 

Let $n$ be a prime. The projective special linear group $\PSL_2(n)$ 
consists of all  
permutations of the set $\{\infty\} \cup \gf(n)$
of the following form: 
$$ 
\tau_{(a,b,c,d)}(x)=\frac{ax+c}{bx+d} 
$$ 
with $ad-bc =1$, and the following conventions: 
\begin{itemize}
\item $\frac{a}{0} = \infty $ for all $a \in \gf(n)^*$. 
\item $\frac{\infty a +c}{\infty b +d} =\frac{a}{b}$. 
\end{itemize} 
The set of all such permutations is a group under
the function composition operation. 
It is known that $\PSL_2(n)$ is generated 
by the following two permutations \cite[p. 491]{MS77}:  
\begin{eqnarray*}
&& S: \ y \mapsto y+1, \\
&& T: \ y \mapsto - \frac{1}{y}. 
\end{eqnarray*}
The permutation group $\PSL_2(n)$ has a number of interesting properties, 
and has applications in both mathematics and engineering. 

Let $\C$ be a cyclic code of length $n$ over $\gf(q)$, where $q$ is a prime power with $\gcd(n, q)=1$. 
Let $\overline{\C}$ denote the extended code of $\C$, where $\infty$ is used to index the extended 
coordinate and other coordinates are indexed by the elements of $\gf(n)$. For any codeword 
$\overline{c}=(c_0, c_1, \cdots, c_{n-1}, c_{\infty})$ in $\overline{\C}$, any permutation $\tau$ 
of $\{\infty\} \cup \gf(n)$ acts on $\overline{c}$ as follows: 
\begin{eqnarray*}
\tau(\overline{c}) = (c_{\tau(0)}, c_{\tau(1)}, \cdots, c_{\tau(n-1)}, c_{\tau(\infty)}). 
\end{eqnarray*}   
The extended code $\overline{\C}$ is said to be invariant under $\PSL_2(n)$ if 
$$ 
\PSL_2(n) (\overline{\C}) = \overline{\C}. 
$$ 
In other words, the extended code $\overline{\C}$ is invariant under $\PSL_2(n)$ if the permutation part of the automorphism group $\GAut(\overline{\C})$ contains $\PSL_2(n)$. 

Linear codes that are invariant under $\PSL_2(n)$ have interesting properties. 
It is well known that the extended odd-like quadratic residue codes are invariant under $\PSL_2(n)$. 
The objective of this paper is to prove that the only such binary codes are the extended odd-like quadratic residue codes and the trivial codes.

\section{Motivations of this paper}

Let $\cP$ be a set of $v \ge 1$ elements, and let $\cB$ be a set of $k$-subsets 
of $\cP$, where $k$ is
a positive integer with $1 \leq k \leq v$.
 Let $t$ be a positive integer with $t \leq k$. The pair
$\bD = (\cP, \cB)$ is called a $t$-$(v, k, \lambda)$ {\em design\index{design}}
\cite{BJL},
or simply {\em $t$-design\index{$t$-design}}, 
if every $t$-subset of $\cP$ is 
contained in exactly $\lambda$ elements of
$\cB$. The elements of $\cP$ are called {\it points}, and those of $\cB$ are
 referred to as {\it blocks}.
We use $b$ to denote the number of blocks in $\cB$.  A $t$-design is 
called {\em simple\index{simple}} if $\cB$ does not contain any repeated blocks.
  A $t$-design is called {\em symmetric\index{symmetric design}} if $v = b$.
 It is clear that $t$-designs with $k = t$ or $k = v$ always exist. Such $t$-designs are {\em trivial}. 
A $t$-$(v,k,\lambda)$ design is referred to as a
 {\em Steiner system\index{Steiner system}} if $t \geq 2$ and $\lambda=1$,
 and is denoted by $S(t,k, v)$.

Let $\C$ be a $[v, \kappa, d]$ linear code over $\gf(q)$. Let $A_i:=A_i(\C)$  
be the
number of codewords with Hamming weight $i$ in $\C$, where $0 \leq i \leq v$.
 The sequence 
$(A_0, A_1, \cdots, A_{v})$ is
called the \textit{weight distribution} of $\C$, 
and $\sum_{i=0}^v A_iz^i$ is referred to as
the \textit{weight enumerator} of $\C$. For each $k$ with $A_k \neq 0$, let $\cB_k$ denote
the set of supports of all codewords of Hamming weight $k$ in $\C$, where the coordinates of a codeword
are indexed by $(0,1,2, \cdots, v-1)$. Let $\cP=\{0, 1, 2, \cdots, v-1\}$.  The pair $(\cP, \cB_k)$
may be a $t$-$(v, k, \lambda)$ design for some
positive  integer $\lambda$ and appropriate $t\ge 0 $, and is called a 
design supported by the code $\C$. 
In this case, we say that $\C$ holds a $t$-$(v, k, \lambda)$ design. 
If $\C$ holds a $t$-$(v,k,\lambda)$ design, its dual code
$\C^{\perp}$ admits majority-logic decoding
up to a certain number of errors determined by the design
parameters
\cite{R}, \cite{RB}, \cite[Section 8]{Tonchev}.

A classical approach to obtain $t$-designs from linear codes is by 
using the automorphism 
groups of linear codes (\cite{AK92,HP03,MS77,Tonchev,Tonchevhb}). 
A proof of the following theorem can be found in \cite[p. 308]{HP03}. 
 
\begin{theorem}\label{thm-designCodeAutm}
Let $\C$ be a linear code of length $n$ over $\gf(q)$ where $\GAut(\C)$ is $t$-transitive. Then the codewords of any weight $i \geq t$ of $\C$ hold a $t$-design.
\end{theorem} 


The following theorem can be derived directly from Propositions 4.6 and 
4.8 in \cite{BJL} (see also \cite[1.27]{Tonchevhb}). 

\begin{theorem}\label{thm-designCodeAutmhomo}
Let $\C$ be a linear code of length $n$ over $\gf(q)$ where $\GAut(\C)$ is $t$-homogeneous. Then the codewords of any weight $i \geq t$ of $\C$ hold a $t$-design.
\end{theorem}

The two theorems above give a sufficient condition for a linear code to hold $t$-designs. To apply Theorems 
\ref{thm-designCodeAutm} and \ref{thm-designCodeAutmhomo}, we need to determine the automorphism group of $\C$ and show that it is $t$-transitive or $t$-homogeneous. 

The projective special linear group $\PSL_2(n)$ is $2$-transitive for all primes $n$ and $3$-homogeneous for primes $n \equiv 3 \pmod{4}$ on 
the set $\{0,1, \cdots, n-1, \infty\}$. It is known that the 
extended odd-like quadratic residue codes are invariant under $\PSL_2(n)$. 
Hence, the 
extended quadratic residue codes hold an infinite family of $2$-designs for $n \equiv 1 \pmod{4}$, 
an infinite family of $3$-designs for $n \equiv 3 \pmod{4}$. To construct more $t$-designs with 
$t \in \{2, 3\}$, one would search for other extended cyclic codes over finite fields that are 
invariant under the action of $\PSL_2(n)$
(cf. \cite{HP03}, \cite{MS77}). This is the main motivation of this paper.

\section{Binary linear codes invariant under $\PSL_2(n)$}

\subsection{All binary linear codes invariant under $\PSL_2(n)$ are extended cyclic codes} 

\begin{theorem}\label{thm-extendedcycliccode} 
Let $n$ be an odd prime. Let $\widetilde{\calC}$
 be a binary linear code of length $n+1$ 
and $\widetilde{\calC} \neq \gf(2)^{n+1}$.
 If $\widetilde{\calC}$ is invariant under 
$\PSL_2(n)$, then $\widetilde{\calC}$ is an extended cyclic code.
\end{theorem} 

\begin{proof} 
If $\widetilde{\calC}$ has dimension $0$ or $n+1$
%
the conclusion is obviously true. 
We now assume that 
\begin{eqnarray}\label{eqn-assumpt}
1 \leq \dim\left(\widetilde{\calC} \right) \leq n-1. 
\end{eqnarray} 
Let $\calC$ be the punctured code of $\widetilde{\calC}$ at coordinate $\infty$.
Since the permutation $\tau_i(x) = x+i$ in $\PSL_2(n)$ acts cyclically on the coordinates $(0,1,...,n-1)$ when $i=1$, ${\calC}$ must be a cyclic code. 

Since $\dim\left(\widetilde{\calC} \right)\geq 1$, 
any minimum weight codeword in $\widetilde{\calC}$ 
has weight at least one. Suppose $\widetilde{\calC}$ has a codeword $c$ with Hamming weight $1$. 
If 
$$ 
c=(0,0, \cdots, 0, 1) \in \widetilde{\calC}, 
$$  
then the permutation 
$$ 
T(x)=-\frac{1}{x} 
$$ 
will transform $c$ into the  codeword 
$$ 
T(c)=(1,0,\cdots, 0, 0)
$$ 
in $\widetilde{\calC}$. Note that the permutation $\tau_i(x) = x+i$ in $\PSL_2(n)$ 
will transform 
$  
T(c) 
$ 
into the following codeword   
$$ 
(0, \cdots, 0, 1, 0, \cdots, 0, 0),
$$ 
where the nonzero bit $1$ could be in any coordinate $i$ with $0 \leq i \leq n-1$. 
This means that all codewords of weight $1$ are in $\widetilde{\calC}$. Consequently, 
$\widetilde{\calC}=\gf(q)^{n+1}$. This is contrary to the assumption that $\widetilde{\calC}\neq 
\gf(2)^{n+1}$. 
This proves that the minimum distance $d(\widetilde{\calC})>1$. It then follows from 
Theorem 1.5.1 in \cite{HP03} that 
$$ 
\dim(\calC)=\dim(\widetilde{\calC}). 
$$
Next we prove that $\widetilde{\calC}$ is the extended code of ${\calC}$.

Let $g(x)=\sum_{i=0}^{n-1}a_ix^i$ be the generator polynomial of ${\calC}$,
 and let $\deg(g)=n-1-k$, i.e. $\dim(\calC)=k$.
 Then the first $k$ cyclic shifts of the codeword
 ${\ba_0} = (a_0,a_1,..,a_{n-1})$ form a basis of ${\calC}$.
 Denote these first $k$ cyclic shifts by $\{{\ba_j} ,~0\le j\le k-1\}$, and their corresponding codewords in $\widetilde{\calC}$ as $\{\widetilde{\ba_j} ,~0\le j\le k-1\}$, which form a basis of $\widetilde{\calC}$. Then we must have $\widetilde{\ba_j} = \tau_j(\widetilde{\ba_0})$ for $0\le j\le k-1$, where 
$\tau_j(\infty)=\infty$, which shows that the $\infty$ coordinate of $\widetilde{\ba_j}$ is the same for $0\le j \le k-1$.

Let $\widetilde{\ba_0}=(a_0,a_1,...,a_{n-1},a_{\infty})$.
If $g(1)\neq 0$, the extended coordinate of $\ba_0$ should be $1$. Assume that $\widetilde{\calC}$ is not the extended code of ${\calC}$, which is equivalent to $a_\infty = 0$.  Since the $\infty$ coordinate of $\widetilde{\ba_j}$ is the same for $0\le j \le k-1$, the $\infty$ coordinate of 
all codewords in $\widetilde{\calC}$ is $0$. However, as $\PSL_2(n)$ is transitive, there exists a permutation and a codeword that transfer a $1$ in the codeword to the $\infty$ coordinate, which gives a contradiction. 
Thus, in this case, $\widetilde{\calC}$ is the extended code of ${\calC}$.

If $g(1)=0$, then $\calC$ is an even-weight code and all extended coordinates
 in its extended code should be $0$. 
Assume $\widetilde{\calC}$ is not the extended code
 of ${\calC}$, which is equivalent to 
$a_\infty = 1$. Since $n$ is odd and $(a_0,a_1,...,a_{n-1})$ has even weight, there exists  
an integer $i$ such that $a_i=0$, where $0 \leq i \leq n-1$.
 Since $\PSL_2(n)$ acts transitively on 
$\{0,1, \cdots, n-1, \infty\}$, 
there must be a permutation $\tau \in \PSL_2(n)$ 
that exchanges coordinate $i$ with coordinate $\infty$. We have then 
$$ 
\tau(a_0,a_1,...,a_{n-1},a_{\infty})=(a_0, \cdots, a_{i-1}, 1, a_{i+1}, \cdots, a_{n-1}, 0)
$$   
which is another codeword in $\widetilde{\C}$. It then follows that 
$$ 
(a_0, \cdots, a_{i-1}, 1, a_{i+1}, \cdots, a_{n-1}) \in \C, 
$$ 
which has odd weight. This is contrary to our assumption that $\calC$ has only even weights. 

\end{proof}

\subsection{The main theorem} 

The main result of this paper is the following. 

\begin{theorem}\label{thm-mymain}
Let $n$ be an odd prime. If $\widetilde{\C}$ is a binary code of length
 $n+1$ invariant under 
$\PSL_2(n)$, then $\widetilde{\C}$ must be one of the following: 
\begin{enumerate}
\item the zero code $\C(0)=\{(0,0, \cdots, 0)\}$; or  
\item the whole space $\C(n+1)$, which is the dual of $\C(0)$; or 
\item the code $\C(1)=\{(0,0, \cdots, 0), (1,1, \cdots, 1)\}$ of
 dimension $1$; or 
\item the code $\C(1)^\perp$, denoted by $\C(n)$, given by 
      $$
      \C(n)=\left\{(c_0, c_1, \cdots, c_n) \in \gf(2)^{n+1}: \sum_{i} c_i =0 \right\}; 
      $$
      or 
\item the extended code of one of the two odd-like quadratic residue binary codes of length $n$.       
\end{enumerate}  
\end{theorem}

According to Theorem 
\ref{thm-extendedcycliccode}, 
to prove Theorem \ref{thm-mymain}, we need to consider only extended 
cyclic codes.
 Before proving Theorem \ref{thm-mymain}, we need do some preparations. 
Specifically, we will make use of the defining set of a cyclic code, 
and the Fourier transform 
(also called the Mattson-Solomon polynomial) of a codeword. 

Note that $\gcd(2, n)=1$. Let $m$ denote the order of $2$ modulo $n$.
 The $2$-cyclotomic 
coset $\bC_i$ modulo $n$ containing $i$ is defined by 
$$ 
\bC_i=\{i, i2, i2^2, \cdots, i 2^{\ell_i-1} \} \bmod{n}, 
$$ 
where $\ell_i$ is the least positive integer such 
that $i2^{\ell_i} \equiv i \pmod{n}$.
 The smallest non-negative integer in $\bC_i$ is called the coset 
leader of $\bC_i$. Let $\Gamma_{(2, n)}$ denote the set of all coset leaders of the $2$-cyclotomic 
cosets modulo $n$. Then $\{\bC_i: i \in \Gamma_{(2, n)}\}$
is 
a partition of the set $\Z_n =\{0, 1, \cdots, n-1\}$.
 We identify $\Z_n$ with $\gf(n)$. 

Let $\alpha$ be a generator of $\gf(2^m)^*$, and let $\beta=\alpha^{(2^m-1)/n}$. Then $\beta$ 
is a $n$-th primitive root of unity in $\gf(2^m)$. It is straightforward to see that the 
minimal polynomial $\m_{\beta^i}(x)$ over $\gf(2)$ of $\beta^i$ is given by 
\begin{eqnarray}\label{eqn-march3}
\m_{\beta^i}(x)=\prod_{j \in \bC_i} (x- \beta^j). 
\end{eqnarray}
Clearly, 
$$ 
x^n-1=\prod_{i \in \Gamma_{(2, n)}} \m_{\beta^i}(x). 
$$ 

The generator polynomial $g(x)$ of any cyclic code $\C$ over $\gf(2)$ of length $n$ must be the 
product of some of irreducible polynomials $\m_{\beta^i}(x)$. The set 
$$ 
T=\{0 \leq i \leq n-1: g(\beta^i)=0\} 
$$ 
is called the defining set of the cyclic code $\C$ with respect to $\beta$, and must be the union 
of some $2$-cyclotomic cosets. 

\par 
\vspace*{.2cm}
The Fourier 
transform of a vector $c=(c_0, c_1, \cdots, c_{n-1}) \in \gf(2)^n$,  
denoted by $C=(C_0, C_1, \cdots, C_{n-1}) \in \gf(2^m)^n$, is given by 
$$ 
C_j=\sum_{i=0}^{n-1} \beta^{ij}c_i = c(\beta^j), 
$$ 
where $c(x)=\sum_{i=0}^{n-1} c_i x^i \in \gf(2)[x]$ and $0 \leq j \leq n-1$.

Let $\calC$ be a cyclic code of prime length $n$ and $\pi$ be a primitive element of $\gf(n)$.
 Then indices in $\gf(n)$ can be expressed by powers of $\pi$,
 and codewords of $\C$ can be reordered accordingly as  
$$ 
c=(c_0, c_{\pi^0}, c_{\pi^1}, \cdots, c_{\pi^{n-2}}). 
$$ 
Note that $\pi^{-1}$ is another generator of $\gf(n)^*$. Similarly, the Fourier transform 
$C$ of $c$ can be written in the permuted order,
$$
C=(C_0, C_{\pi^{-0}}, C_{\pi^{-1}}, \cdots, C_{\pi^{-(n-2)}}). 
$$

Rewrite the Fourier transform as $C_0=\sum_{i=0}^{n-1} c_i$ { and } $
C_j = c_0 + \sum_{i=1}^{n-1} \beta^{ij} c_i, \  j=1, 2, \cdots, n-1. 
$
With the index changing described above, the second equation can be expressed as 
\begin{equation}
C_{\pi^{-s}}=c_0 + \sum_{r=0}^{n-2} \beta^{\pi^{r-s}} c_{\pi^r}, \ \ \  s=0,1, \cdots, n-2. 
\end{equation} 
Let $C'_s = C_{\pi^{-s}}$ and $c'_r = c_{\pi^r}$ for $0\le s,r\le n-2$. Then 
\begin{equation}\label{eqn-jj152}
C'_{s}=c_0 + \sum_{r=0}^{n-2} \beta^{\pi^{r-s}} c'_{r}, \ \ \  s=0,1, \cdots, n-2.
\end{equation} 
This can be rewritten in the language of polynomials.
Define 
\begin{eqnarray*}
u(x) =  \sum_{r=0}^{n-2} \beta^{\pi^{-r}} x^r, \ \ 
c'(x) = \sum_{r=0}^{n-2} c'_rx^r, \ \ 
C'(x) = \sum_{r=0}^{n-2} C'_rx^r. 
\end{eqnarray*} 
Then all the equations in (\ref{eqn-jj152}) can be compactly expressed into 
\begin{equation}
C'(x)= \left(u(x)c'(x) + c_0 \sum_{i=0}^{n-2} x^i \right) \bmod{(x^{n-1}-1)}, 
\end{equation} 
which is a polynomial representation of the equation of the Fourier transform. 

Let $h = (n-1)/m$. Since $m$ is the order of $2$ modulo $n$, there exists a primitive element $\pi$ of $\gf(n)$ for which $\pi^h=2$ in $\gf(n)$. Then the nonzero elements of $\gf(n)$ can be presented as $\{\pi^r\,:\,0\le r\le n-2\}$. We denote the set of quadratic residues in it by $\cQ := \{\pi^{2r}\,:\,0\le r\le (n-3)/2\}$ and the set of quadratic nonresidues by $\cN$.

Recall that $\beta=\alpha^{(2^m-1)/n}$, which is a $n$-th primitive root of unity in $\gf(2^m)$, 
where $\alpha$ is a generator of $\gf(2^m)^*$.  We now prove the following lemma. 

\begin{lemma}\label{lemmaBeta} 
Let notation and symbols be as before. Define $\beta_l=\beta^l$ for $1 \leq l \leq n-1$. 
Then $\{1, \beta_l, \cdots, \beta_l^{m-1}\}$ is a basis of 
$\gf(2^m)$ over $\gf(2)$. Consequently, for any element $a \in \gf(2^m)$, there exists a 
polynomial $f(x)$ over $\gf(2)$ with degree less than $m$ such that $f(\beta^l) = a$.
\end{lemma} 

\begin{proof}
Let $\bC_l=\{l, l2, \cdots, l2^{m-1}\} \bmod n$ be the $2$-cyclotomic coset modulo $n$ containing $l$, 
where $1 \leq l \leq n-1$. Recall that $n$ is a prime. By definition, $m=\ord_n(2)$. Let $w$ be a 
positive integer such that $l2^w \equiv l \pmod{n}$. Then $l(2^w-1) \equiv 0 \pmod{n}$. Since $n$ 
is a prime and $1 \leq l \leq n-1$, $2^w \equiv 1 \pmod{n}$. It then follows from $m=\ord_n(2)$ that 
$w \geq m$. Consequently, $|\bC_l|=m$. 
Hence, the polynomial $\m_{\beta^l}(x)$ 
of (\ref{eqn-march3}) has degree $m$ and is irreducible. This is the minimal polynomial of $\beta^l$ 
over $\gf(2)$. It then follows that $\{1, \beta_l, \cdots, \beta_l^{m-1}\}$ is a basis of 
$\gf(2^m)$ over $\gf(2)$. Consequently, for any element $a \in \gf(2^m)$, there exists a 
polynomial $f(x)$ over $\gf(2)$ with degree less than $m$ such that $f(\beta^l) = a$. 
\end{proof}

Let $\C$ be a cyclic code over $\gf(2)$ with length $n$ and defining set $T$. Denote the extended code of $\calC$ by $\overline{\calC}$, where the extended coordinate $c_{\infty}$ is defined by 
$$c_\infty = \sum_{i=0}^{n-1}c_i$$
for any codeword $c = (c_0, c_1,...,c_{n-1})\in\calC$.

Consider now the permutation $T:  y \to - 1/y$ in $\PSL_2(n)$. For any  
$ 
\bar{c}=(c_0, c_1, \cdots, c_{n-1}, c_{\infty}) \in \overline{\C}, 
$ 
let $\bar{d}=(d_0, d_1, \cdots, d_{n-1}, d_{\infty})$ be the permuted vector of $c$ under $T$. 
Let $C$ and $D$ be the Fourier transforms of $(c_0, c_1, \cdots, c_{n-1})$ and $(d_0, d_1, 
\cdots, d_{n-1})$, respectively. Define the polynomials $D'(x) = \sum_{s=0}^{n-2}D'_sx^s$. 
We have the following relationship between $D'(x)$ and $C'(x)$, which was
 stated in \cite{Blahut91} 
without a proof. We state it as a general result here and present a proof.  

\begin{lemma}\label{lammaBlahud}
Let $D'(x)$, $C'(x)$, $u(x)$ be defined as above. We have 
\begin{equation}
D'\left(\frac{1}{x}\right) = u(x)^2 C'(x) \pmod{x^{n-1}-1}.
\end{equation}
\end{lemma}

\begin{proof}

The inverse Fourier transform can be written as 
$$ 
c_i=C_0 + \sum_{k=1}^{n-1} \beta^{-ik} C_k= c_{\infty} + \sum_{k=1}^{n-1}  \beta^{-ik} C_k. 
$$ 
Now we have 
$$ 
d_i=c_{-1/i}=c_{\infty} + \sum_{k=1}^{n-1}  \beta^{(1/i)k} C_k, \ \ \  i=1,2, \cdots, n-1,  
$$
and $d_0=c_{\infty}$. Consequently, for $1\le j\le n-1$,
\begin{eqnarray}\label{eqn-jj153}
D_j &=& d_0 + \sum_{i=1}^{n-1} \beta^{ij} d_i \nonumber \\ 
&=& c_{\infty} + \sum_{i=1}^{n-1} \beta^{ij} \left( c_{\infty} + \sum_{k=1}^{n-1}  \beta^{(1/i)k} C_k  \right) \nonumber \\
&=& \sum_{i=1}^{n-1} \beta^{ij} \sum_{k=1}^{n-1}  \beta^{(1/i)k} C_k. 
\end{eqnarray} 
Now we change indices again as follows: 
$$ 
i=\pi^r, \ \ k=\pi^t, \ \ j=\pi^{-s}, 
$$ 
with $r,s,t\in \mathbb{Z}_{n-1}$.
Then (\ref{eqn-jj153}) becomes 
\begin{equation}\label{eqn-jj154}
D_{\pi^{-s}} = \sum_{r=0}^{n-2} \beta^{\pi^{-s+r}} \sum_{t=0}^{n-2}  \beta^{\pi^{-r+t}} C_{\pi^t}.
\end{equation} 
Alternatively, (\ref{eqn-jj154}) can be expressed as 
\begin{equation}\label{eqn-jj155}
D'_{-s} = \sum_{r=0}^{n-2}  u_{s-r} \sum_{t=0}^{n-2}  u_{r-t} C'_{t}, 
\end{equation} 
where $C'_{t}=C_{\pi^t}$, $D'_{-s}=D_{\pi^{-s}}$, $u_r=\beta^{\pi^{-r}}$  and $r,s,t\in \mathbb{Z}_{n-1}$.
Thus we have
\begin{align*}
u(x)^2 C'(x) &= \left(\sum_{s=0}^{n-2}u_sx^s\right)\left(\sum_{r=0}^{n-2}u_rx^r\right)
               \left(\sum_{t=0}^{n-2}C'_tx^t \right) \\
&= \left(\sum_{s=0}^{n-2}u_sx^s \right) \sum_{t=0}^{n-2}~\sum_{r=0}^{n-2}(u_rC'_t)x^{r+t}\\
&= \left(\sum_{s=0}^{n-2}u_sx^s \right) 
   \sum_{t=0}^{n-2}~\sum_{r'=t}^{t+n-2}u_{r'-t}C'_tx^{r'} & \\
&= \left(\sum_{s=0}^{n-2}u_sx^s\right) 
   \sum_{t=0}^{n-2}~\sum_{r=0}^{n-2}u_{r-t}C'_tx^{r} \pmod{x^{n-1}-1} \\
&=\sum_{t=0}^{n-2}~\sum_{r=0}^{n-2}u_{r-t} C'_t\sum_{s=0}^{n-2}u_sx^{r+s} \pmod{x^{n-1}-1} \\
&=\sum_{t=0}^{n-2}~\sum_{r=0}^{n-2}u_{r-t}C'_t\sum_{s'=s}^{s+n-2}u_{s'-r} x^{s'} \pmod{x^{n-1}-1} \\
&=\sum_{t=0}^{n-2}~\sum_{r=0}^{n-2}u_{r-t}C'_t\sum_{s=0}^{n-2}u_{s-r}x^{s} \pmod{x^{n-1}-1} \\
&=\sum_{s=0}^{n-2}\left(\sum_{r=0}^{n-2}u_{s-r}\sum_{t=0}^{n-2}u_{r-t} C'_t\right)x^{s} \pmod{x^{n-1}-1} \\
&=\sum_{s=0}^{n-2}D'_{-s}x^{s} \pmod{x^{n-1}-1}\\
&=D'(1/x) \pmod{x^{n-1}-1}.
\end{align*}
This completes the proof. 
\end{proof}

We are now ready to prove the following. 

\begin{theorem}\label{mainprop}
Let $\C$ be a binary cyclic code of length $n$ with defining set $T$. Assume the extended code 
$\overline{\C}$ is invariant under $\PSL_2(n)$. If there exists an $l\in\cQ$ ($\cN$, respectively) 
that is not in $T$, then $\cQ\cap T = \emptyset$ ($\cN\cap T = \emptyset$, respectively). Further, 
if $0 \in T$, then $\C$ must be the zero code. 
\end{theorem}

\begin{proof}
Let $\bc = (c_0,c_1...c_{n-1}) \in\calC $ be a codeword of $\calC$ and let the corresponding polynomials $C'(x), ~D'(x)$ and $u(x)$ be the same as before. Lemma \ref{lammaBlahud} says that 
\begin{equation}\label{maineq}
D'(1/x) = \sum_{r=0}^{n-2}D'_rx^{n-1-r} = u^2(x)C'(x) \bmod{(x^{n-1}-1)},
\end{equation} 
where $u(x) = \sum_{r=0}^{n-2}\beta^{\pi^{-r}}x^r$. 

We first consider the case that there exists an $l=\pi^{2u}\in\cQ$ such that $l\notin T$, 
where $0 \leq u \leq (n-3)/2$. To show that $T\cap\cQ=\emptyset$, we need to prove the following statement. \\

\noindent 
\textit{Statement:} For any even number $2s$ with $0\le 2s\le n-2$, there exists a codeword $c(s) \in\calC$ such that the corresponding term $D'_{2s}x^{n-1-2s}$ in Equation (\ref{maineq}) is nonzero, i.e. $D'_{2s} \neq 0 $. \\ 

\noindent 
\textit{Proof of the Statement:} Let $\overline{\m}_{\beta^i}(x) =(x^n-1)/\m_{\beta^i}(x)$ be a polynomial in $\gf(2)[x]$. For any $0\le s\le (n-3)/2$, define a codeword $\bc\in\calC$ by the following polynomial.
$$c(x) = a(x)\cdot\overline{\m}_{\beta^l}(x),$$
where $a(x)$ is a polynomial of degree less than $m$ over $GF(2)$, which will be figured out later.
Since $l\notin T$, clearly $c(x)$ is a codeword in $\calC$ with $C_i = c(\beta^i) =0$ for any $i\notin \bC_l$. Besides, for $i = l\cdot 2^w\in \bC_l$ it is easy to see that $C_{2^w\cdot l}=(C_{l})^{2^w}$ for $0\le w\le m-1$. Equivalently,
\begin{equation*}
C'_r = C_{\pi^r} = \begin{cases}
(C'_{2u})^{2^w} = C_l^{2^w},&~r = 2u+hw,~0\le w\le m-1;\\
0,&~otherwise.
\end{cases}
\end{equation*} 
From Equation (\ref{maineq}) we have 
\begin{align}\label{maineq2}
D'(1/x) &= g^2(x)C'(x) \pmod{x^{n-1}-1} \nonumber \\
&=\left(\sum_{r=0}^{n-2}\beta^{\pi^{-r}}x^r\right)^2\cdot \left(\sum_{w=0}^{m-1}C'_{2u+hw} x^{2u+hw}\right)  \pmod{x^{n-1}-1} \nonumber \\
&=x^{2u}\cdot\left(\sum_{r=0}^{n-2}\beta^{2\pi^{-r}}x^{2r}\right)\cdot
\left(\sum_{w=0}^{m-1}C_l^{2^w}x^{hw}\right)  \pmod{x^{n-1}-1} 
\end{align}
We continue our discussion by distinguishing the following two circumstances. First we consider the case that $2=\pi^h$ is a quadratic number in $\gf(n)$. Let $h'=h/2$. It then follows from (\ref{maineq2}) that 
\begin{align*}
& D'(1/x) \\ 
&=x^{2u}\sum_{r=0}^{n-2}\sum_{w=0}^{m-1}\beta^{2\pi^{-r}}C_l^{2^w}x^{2r+2h'w}  \pmod{x^{n-1}-1}\\
&=x^{2u}\sum_{w=0}^{m-1}C_l^{2^w}\sum_{s=h'w}^{h'w+n-2}(\beta^{2\pi^{h'w}})^{\pi^{-s}}x^{2s}  \pmod{x^{n-1}-1} \\
&=x^{2u}\sum_{w=0}^{m-1}C_l^{2^w}\sum_{s=0}^{n-2}(\beta^{2\pi^{h'w}})^{\pi^{-s}}x^{2s}   \pmod{x^{n-1}-1} \\
&=x^{2u}\sum_{s=0}^{n-2}\left(\sum_{w=0}^{m-1}C_l^{2^w}\beta^{2\pi^{h'w-s}}\right)x^{2s} 
      \pmod{x^{n-1}-1} \\
&=x^{2u}\sum_{s=0}^{\frac{n-3}{2}}\left(\sum_{w=0}^{m-1}C_l^{2^w}\beta^{2\pi^{h'w-s}}\right)x^{2s} + x^{2u}\sum_{s=\frac{n-1}{2}}^{n-2}\left(\sum_{w=0}^{m-1}C_l^{2^w}\beta^{2\pi^{h'w-s}}\right)x^{2s} 
      \pmod{x^{n-1}-1}\\
&=x^{2u}\sum_{s=0}^{\frac{n-3}{2}}\left(\sum_{w=0}^{m-1}C_l^{2^w}\beta^{2\pi^{h'w-s}}\right)x^{2s} + x^{2u}\sum_{s=0}^{\frac{n-3}{2}}\left(\sum_{w=0}^{m-1}C_l^{2^w}\beta^{2\pi^{h'w-s-\frac{n-1}{2}}}\right)x^{2s+n-1}     \pmod{x^{n-1}-1}\\
&=x^{2u}\sum_{s=0}^{\frac{n-3}{2}}\left(\sum_{w=0}^{m-1}(\beta^{2\pi^{h'w-s}}+\beta^{-2\pi^{h'w-s}})C_l^{2^w}\right)x^{2s}  \pmod{x^{n-1}-1}.
\end{align*} 
Thus for any $0\le s\le (n-3)/2$, by comparing terms of both sides we have 
$$D'_{2s} = \sum_{w=0}^{m-1}\left(\beta^{2\pi^{h'w+s+u-\frac{n-1}{2}}}+\beta^{-2\pi^{h'w+s+u-\frac{n-1}{2}}}\right)C_l^{2^w}= \sum_{w=0}^{m-1}\left(\beta^{2\pi^{h'w+s+u}}+\beta^{-2\pi^{h'w+s+u}}\right)C_l^{2^w}$$

Define a polynomial $L_s(x)$ over $\gf(2^m)$ to be 
$$L_s(x) = \sum_{w=0}^{m-1}\left(\beta^{2\pi^{h'w+s+u}}+\beta^{-2\pi^{h'w+s+u}}\right)x^{2^w}.$$ 
$L_s(x)$ is a linearized function independent of $\bc$ and $D'_{2s} = L_s(C_l)$. 
We now prove that $L_s(x)$ is not the zero function for any $s$. Notice that the degree 
of $L_s(x)$ is at most $2^{m-1}$. It suffices to prove that 
$$
\beta^{2\pi^{h'w+s+u}}+\beta^{-2\pi^{h'w+s+u}} \neq 0 
$$ 
for one $w$ with $0 \leq w \leq m-1$. We do this for $w=0$. On the contrary, suppose that 
$$ 
\beta^{2\pi^{h'w+s+u}}+\beta^{-2\pi^{h'w+s+u}} = 0 
$$ 
for $w=0$. We have then $\beta^{2\pi^{s+u}}+\beta^{-2\pi^{s+u}} = 0$, which is the same as 
$\beta^{4\pi^{s+u}}=1$. Note that $x \mapsto x^4$ is a permutation on $\gf(2^m)$. We obtain 
that $\beta^{\pi^{s+u}}=1$. Obviously, $\pi^{s+u} \bmod{n}$ is an integer in the set 
$\{1, 2, \cdots, n-1\}$. Since $\beta$ is an $n$-th primitive root of unity in $\gf(2^m)$, 
$\beta^{\pi^{s+u}} \neq 1$. Thus, we have reached a contradiction. 
Therefore, $L_s(x)$ is a nonzero function. 
Consequently, there is an element 
$\gamma$ in $\gf(2^m)$ such that $L_s(\gamma)\neq 0$. By Lemma \ref{lemmaBeta}, there exists a polynomial $a_s(x) \in \gf(2)[x]$ of degree less than $m$ such that $a_s(\beta^l) = {\gamma}{\overline{\m_{\beta^l}}(\beta^l)}^{-1}$. Set the $a(x)$ in the definition of $c(x)$ to be $a_s(x)$. 
Then we have 
$$C'_{2u} = C_l = c(\beta^l) = a_s(\beta^l)\cdot \overline{\m}_{\beta^l}(\beta^l) =\gamma.$$
Thus $D'_{2s} = L_s(C'_{2u})=L_s(\gamma)\neq 0$. 

Finally, we consider the case that $2=\pi^h\in\cN$. Then $h$ is odd and $m=(n-1)/h$ is even. Equation (\ref{maineq2}) becomes 
\begin{align*}
& D'(1/x) \\ 
&=x^{2u}\left(\sum_{r=0}^{n-2}\beta^{2\pi^{-r}}x^{2r}\right)\cdot
\left(\sum_{w=0}^{m/2-1}C_l^{2^{2w}}x^{2hw}+C_l^{2^{2w+1}}x^{2hw+h}\right) \pmod{x^{n-1}-1} \\
&=x^{2u}\sum_{w=0}^{m/2-1}\sum_{r=0}^{n-2}\beta^{2\pi^{-r}}x^{2r}\cdot
\left(C_l^{2^{2w}}x^{2hw}+C_l^{2^{2w+1}}x^{2hw+h}\right) \pmod{x^{n-1}-1} \\
&=x^{2u}\sum_{w=0}^{m/2-1}\sum_{r=0}^{n-2}\beta^{2\pi^{-r}}
C_l^{2^{2w}}x^{2(hw+r)}+x^{2u+h}\sum_{w=0}^{m/2-1}\sum_{r=0}^{n-2}\beta^{2\pi^{-r}}C_l^{2^{2w+1}}x^{2(hw+r)} \pmod{x^{n-1}-1}  \\
&=x^{2u}\sum_{w=0}^{m/2-1}C_l^{2^{2w}}\sum_{s=hw}^{hw+n-2}\beta^{2\pi^{hw-s}}
x^{2s}+x^{2u+h}\sum_{w=0}^{m/2-1}C_l^{2^{2w+1}}\sum_{s=hw}^{hw+n-2}\beta^{2\pi^{hw-s}}x^{2s} \pmod{x^{n-1}-1} \\
&=x^{2u}\sum_{w=0}^{m/2-1}C_l^{2^{2w}}\sum_{s=0}^{n-2}\beta^{2\pi^{hw-s}}
x^{2s}+x^{2u+h}\sum_{w=0}^{m/2-1}C_l^{2^{2w+1}}\sum_{s=0}^{n-2}\beta^{2\pi^{hw-s}}x^{2s} \pmod{x^{n-1}-1} \\
&= x^{2u}\sum_{s=0}^{\frac{n-3}{2}}\left(\sum_{w=0}^{m/2-1}(\beta^{2\pi^{h'w-s}}+\beta^{-2\pi^{h'w-s}})C_l^{2^{2w}}\right)x^{2s}\\ 
&\ \ +
x^{2u+h}\sum_{s=0}^{\frac{n-3}{2}}\left(\sum_{w=0}^{m/2-1}(\beta^{2\pi^{h'w-s}}+\beta^{-2\pi^{h'w-s}})C_l^{2^{2w+1}}\right)x^{2s} \pmod{x^{n-1}-1}. 
\end{align*} 
Again by comparison, for any $0\le s\le (n-3)/2$ we have 
$$D'_{2s} = \sum_{w=0}^{m/2-1}\left(\beta^{2\pi^{h'w+s+u}}+\beta^{-2\pi^{h'w+s+u}}\right)C_l^{2^{2w}}.$$
With similar arguments to the case $2\in\cQ$, we can find a $\bc_s\in\calC$ such that $D'_{2s}\neq 0$. We hereby finish the proof of the statement. 

By the statement above, for any $\pi^{2s}\in\cQ$, there exists a codeword $c(s) \in\calC$ such that $d(\beta^{\pi^{2s}}) = D_{\pi^{2s}} = D'_{2s} \neq 0$, where $d(x)=\sum_{i=0}^{n-1}d_ix^i$. Since 
$\bar{d}$ is also a codeword of $\overline{\calC}$, $d(\beta^{\pi^{2s}})\neq 0$ leads to $\pi^{2s} \notin T$. 
Thus we proved that $\cQ \cap T = \emptyset$.

For the case that there exists an $l = \pi^{2u+1} \in \cN$ such that $l \notin T$, the desired 
conclusion can be similarly proved.  

The conclusion of the last part is implied by the proof of Theorem \ref{thm-extendedcycliccode}

\end{proof}

\begin{theorem}\label{thm-pslcodeonlyif}
Let $\calC$ be a binary cyclic code with prime length $n>2$ and $0 <\dim(\calC)<n$.  
Let $\overline{\calC}$ be invariant under $\PSL_2(n)$. Let $g(x)$ denote the generator 
polynomial of $\calC$.  

If $n \equiv \pm 3 \pmod{8}$, then $g(x)=(x^n-1)/(x-1)$. 

If $n \equiv \pm 1 \pmod{8}$, then $g(x)=(x^n-1)/(x-1)$, or $g(x)$ 
is the generator polynomial of one of the odd-like quadratic residue binary codes. 
\end{theorem}

\begin{proof} 
Let $\calC$ be a binary cyclic code with prime length $n>2$ and $0 <\dim(\calC)<n$. 
Let $\overline{\calC}$ be invariant under $\PSL_2(n)$. 
By Theorem \ref{mainprop}, $\dim(\calC) \in \{1, (n+1)/2\}$.

It is well known that $2$ is a quadratic residue modulo $n$ if and only if 
$n \equiv \pm 1 \pmod{8}$. In the case that $n \equiv \pm 1 \pmod{8}$, $2$ 
is a quadratic residue. By 
Theorem \ref{mainprop} we have only the following four possibilities: 
\begin{enumerate} 
\item The defining set $T=\bQ \cup \bN$. In this subcase, $g(x)=(x^n-1)/(x-1)$. 
\item The definition set $T=\bQ$. In this subcase, 
      $$ 
        g(x)=\prod_{i \in \bQ} (x-\beta^i). 
      $$      
\item The definition set $T=\bN$. In this subcase, 
      $$ 
        g(x)=\prod_{i \in \bN} (x-\beta^i). 
      $$         
\end{enumerate} 

Consider now the case $n \equiv \pm 3 \pmod{8}$. In this case, $2$ must be a quadratic nonresidue. 
Hence, $(n-1)/\ord_n(2)$ must be odd. Note that $(x^n-1)/(x-1)$ is the product of $(n-1)/\ord_n(2)$ 
irreducible polynomials of degree $\ord_n(2)$ over $\gf(2)$. Since $(n-1)/\ord_n(2)$ is odd, there is no 
binary cyclic code of length $n$ and dimension $(n + 1)/2$. It then follows that $\C$ must have 
generator polynomial $(x^n-1)/(x-1)$. Another way to prove this conclusion goes as follows. 
If $T$ contains a quadratic residue $a$ modulo $n$, then $T$ must contain $2a \mod{n}$, which 
is a quadratic nonresidue. If $T$ contains a quadratic nonresidue $b$ modulo $n$, then $T$ must 
contain $2b \mod{n}$, which is a quadratic residue. Hence, $\C$ must have defining set 
$T=\{1,2,\cdots, n-1\}$, and thus generator polynomial $(x^n-1)/(x-1)$.      
\end{proof}

\subsection*{\textbf{Proof of Theorem \ref{thm-mymain}:}}  

It is known that the two odd-like quadratic residue binary codes are invariant under 
$\PSL_2(n)$ (see, for example, \cite{Blahut91}). It is easily seen that the four trivial 
binary codes $\C(0), \C(1), \C(n+1)$ and $\C(n)$ of length $n+1$ are invariant under $\PSL_2(n)$. 

Let $\widetilde{\C}$ be a binary code of length $n+1$ which is invariant under $\PSL_2(n)$. 
By Theorem \ref{thm-extendedcycliccode}, $\widetilde{\C}$ is either the whole space $\gf(2)^{n+1}$ 
or an extended cyclic code.   
The desired conclusion 
of the other part then follows from Theorem 
\ref{thm-pslcodeonlyif}.

\section{Concluding remarks} 

The main result of this paper tells us that the only nontrivial binary linear codes 
of length $n+1$ that are invariant under $\PSL_2(n)$, where $n$ is an odd prime, are 
the extended codes of the two odd-like quadratic residue codes of length $n \equiv 
\pm 1 \pmod{8}$. This means that the extended quadratic residue codes are very special. When 
$n \equiv - 1 \pmod{8}$, the extended quadratic residue codes over $\gf(2)$ are self-dual 
and hold $3$-designs. 

A self-dual binary code with parameters $[N, N/2, d]$ is said to be \emph{extremal and of Type II} 
if 
$$ 
d=4\left\lfloor \frac{N}{24} \right\rfloor + 4.  
$$   
It is known that only finitely many Type II extremal codes could exist. In fact, Type II 
extremal codes of length $N$ are known only for the following $N$: 
$$ 
8, 16, 24, 32, 40, 48, 56, 64, 80, 88, 104, 112, 136. 
$$   
Among these Type II extremal codes, those with length in $\{8, 24, 32, 48, 80, 104\}$ 
are extended quadratic residue codes \cite{Malevich}. This fact shows another specialty  
of the quadratic residue codes.  

Experimental data indicates that the main conclusion of this paper (i.e., 
Theorem \ref{thm-mymain}) is also true for linear codes over $\gf(q)$ for 
any prime power $q$. However, even if this is indeed true, it may not be 
easy to prove it. The reader is cordially invited to settle this open  
problem.

\end{document}